\documentclass[nofootinbib,
tightenlines,
superscriptaddress,11pt]{revtex4}

\usepackage{graphicx}
\usepackage{amsmath,amssymb,amsfonts,amsthm,stmaryrd,mathtools,bm,physics}
\usepackage{color}
\usepackage{tikz}
\usepackage[normalem]{ulem}
\usepackage{braket}
\allowdisplaybreaks[1]
\usepackage[bookmarks,linktocpage, colorlinks=true, plainpages = false, citecolor = treegreen,  linkcolor=darkblue, urlcolor = darkblue, filecolor = blue]{hyperref} 

\definecolor{darkblue}{rgb}{0., 0.4, 0.8}
\definecolor{treegreen}{rgb}{0., 0.7, 0.3}

\newtheorem{lemma}{Lemma}

\newcommand{\eps}{\epsilon}
\newcommand{\del}{\partial}
\newcommand{\Curl}{\mathcal{C}}
\newcommand{\D}{\mathcal{D}}
\newcommand{\Lag}{\mathcal{L}}
\let\Box\square

\def\be#1\ee{\begin{align}#1\end{align}}

\def\ba{\begin{eqnarray}}
\def\ea{\end{eqnarray}}
\def\nn{\nonumber}
\def\q{\quad}

\begin{document}

\title{Gravitational wave signatures from area metric gravity}
\author{Bianca Dittrich}
\email{bdittrichATperimeterinstitute.ca} 
\affiliation{Perimeter Institute, 31 Caroline Street North, Waterloo, ON, N2L 2Y5, Canada}

\begin{abstract}
\noindent
Area metric theories have been proposed as effective descriptions in a number of quantum gravity approaches: as continuum limit of effective spin foams, as an effective description for string theory and in the context of holography. 
Area metric actions allow for non-topological parity violating terms, which opens the possibility of observationally determining the Barbero-Immirzi parameter.  We consider the shift symmetric version of linearized area metric theory in Lorentzian signature, which arises in particular from spin foams and modified Plebanski theory, and  avoids ghost and tachyonic instabilities in the graviton sector.   We discuss the gravitational wave like solutions of this theory and through an analysis of the coupling to electro-magnetism, deduce the signals in gravitational wave detectors. Using polarized light in such interferometers allows in principle to measure area metric induced birefringence. For the non-length excitations that accompany a gravitational wave this birefringence turns out to be independent of the Barbero-Immirzi parameter; the parameter is instead encoded in the polarization structure of the massive non-length waves, where it rotates the axes of the induced birefringence. Assuming Planck massive non-length excitations in the area metric the resulting signals are however extremely weak.
\end{abstract}

\maketitle

\section{Introduction}\label{sec:intro}

A key open question in quantum gravity is what are the fundamental gravitational degrees of freedom. Whereas different approaches offer different answers for the deep ultraviolet regime \cite{deBoer:2022zka,Buoninfante:2024yth}, \emph{area metrics} are proposed across a number of approaches to describe the dynamics at intermediate scales between the deep ultraviolet and the infrared regime. 
Area-metric theories have been suggested to capture phenomenological aspects of string theory~\cite{Schuller:2005yt,Punzi:2006hy,Ho:2015cza,Borissova:2024cpx}, AdS/CFT holography \cite{Bhattacharya:2025hag} and, in particular, loop quantum gravity and spin foams \cite{Dittrich:2021kzs,Dittrich:2022yoo,Borissova:2022clg,BDK,Dittrich:2023ava}.

An area metric is a metric on the space of bivectors: it assigns areas to infinitesimal parallelograms and
dihedral angles between them. It is a rank-$4$ tensor with the same symmetries as the Riemann curvature tensor, and therefore has in four spacetime dimensions 20 independent components.~\footnote{A further generalization to so-called acyclic area metrics has 21 independent components~\cite{Schuller:2005yt}.} A given length metric $g$  induces an area metric $G$ via $
G_{\mu\nu\rho \sigma} = g_{\mu\rho}g_{\nu\sigma} - g_{\mu\sigma}g_{\nu\rho}$. In four dimensions area metrics are a genuine  enlargement of the configuration space of gravity due to the 10 additional degrees of freedom.  Thus the 20 area metric degrees of freedom can be split into a length induced part and a non-length part. On a flat spacetime background the 10 non-length degrees of freedom amount to two massive spin-2 fields, one being self-dual (or chiral) and the other being anti-self-dual (or anti-chiral).

Using an expansion of area metrics on flat spacetime, the works  \cite{Alex:2019ari,BDK} classified the space of diffeomorphism-invariant actions quadratic in area-metric fluctuations and at most second order in derivatives.  There are four additional coupling constants (modulo absorption of couplings into field rescalings), two couplings characterize the coupling of the chiral and the anti-chiral fields to the length metric, and another two providing the masses for the chiral and anti-chiral fields. 

Setting the masses to be equal and imposing a so-called shift-symmetry condition on the remaining two couplings singles out a family of quadratic actions, which avoids an extra spin-2 pole and leads to a ghost-free propagator for the effective action for the length metric perturbations. Within this class of couplings the length-metric sector remains exactly massless and without birefringence while the non-length degrees of freedom acquire a mass $m$. The second coupling parameter, denoted $\xi$ below,
controls the strength of parity violation. It is this family that we analyze in the
present paper.\footnote{In due course we will see that, assuming $m$ to be of Planck scale, small deviations from shift-symmetry will not alter the phenomenology substantially, whereas large deviations are largely ruled out observationally.}


The presence of a parity-violating coupling in the quadratic action is a genuinely \emph{area-metric}
phenomenon. For a length metric, no parity-violating term exists at quadratic order in
the fluctuations and second order in derivatives: the available parity-odd invariants
are topological or total derivatives, so the graviton dynamics of general relativity is
necessarily parity-even. Parity violating effects in gravity have been considered within the framework of higher derivative theories, e.g. Chern-Simons gravity \cite{Jackiw:2003pm,Alexander:2009tp,Crisostomi:2017ugk}, where the parity breaking parameter is promoted to a field; via coupling to fermions \cite{Perez:2005pm,Freidel:2005sn}; or other non-Riemannian extensions of gravity, see \cite{Qiao:2022mln} for a comprehensive review.  A key difference between the shift-symmetric area metric theories and parity violating theories discussed in \cite{Qiao:2022mln}, is that in the former case the propagation of gravitational waves (in the linearized theory) is \emph{not} modified from linearized general relativity, whereas the theories discussed in \cite{Qiao:2022mln} lead to velocity and amplitude birefringence for gravitational waves.  We will however see that within shift symmetric area metric dynamics a gravitational wave is dressed with non-length components, which lead, via a coupling to electro-magnetism,  to birefringence for light.

The inverse of the Barbero-Immirzi parameter \cite{Barbero,Immirzi:1997}, which plays a key role in loop quantum gravity and spin foams \cite{ Rovelli:1997na,Dittrich:2012rj,Meissner:2004ju,BarberoG:2022ixy}, multiplies the quasi-topological and parity-violating Holst term \cite{Holst:1995pc} in first order formulations of gravity. The quasi-topological property  ensures, that the parameter does not appear in the classical equations of motion.  But as we will review in Section \ref{Sec:SF}, there is ample evidence that the configuration space in loop quantum gravity and spin foams is rather given by area metrics than length metrics. Furthermore, an area metric action can be derived from the continuum limit of effective spin foams \cite{Dittrich:2022yoo,AsanteToAppear}, and it does satisfy the shift symmetry condition.  The Barbero-Immirzi parameter can therefore be identified with the parity violating parameter $\xi$ in the area metric actions discussed here.\footnote{$\sinh(2\xi)=\gamma^{-1}$ where $\gamma$ is the Barbero-Immirzi parameter.} One outcome of this
paper is that, the parameter $\xi$ is \emph{in principle observable}: it fixes the relative
phase of the self-dual and anti-self-dual components of the massive non-length waves,
and thereby rotates the axes of the birefringence such waves imprint on light by
$\tfrac12\arctan(\sinh2\xi)$, cf.\ Eq.~\eqref{eq:massbiref}.

Whereas previous work performed a canonical analysis of linearized, shift-symmetric area metric actions \cite{BDK},
 we develop in this paper the Lagrangian analysis and construct plane wave solutions in Section \ref{sec:action} and \ref{sec:eom}. This leads to a very transparent encoding of the dynamics: we find  an exactly massless
graviton branch, whose excitations are dressed by non-length components at relative
order $(k_{\text{grav}}/m)^2$, and a massive branch, that carries no length-metric perturbation at
all. We discuss the energetics of the two branches, which is governed by the indefinite
signature of the non-length sector. We then couple the theory to electromagnetism (Section \ref{sec:EM}))
through the constitutive (pre-metric) form of the Maxwell action \cite{Obukhov:2000nw,Hehl:2004zz}, for which the area
metric is the natural geometric input, and derive the optical response: the length
perturbation produces the familiar polarization-blind time delay, while the non-length
perturbations render the vacuum birefringent. This cleanly separates two observational
channels in a laser interferometer --- the standard strain channel, which remains
unmodified at linear order, and a polarization channel, which vanishes identically in
general relativity and therefore constitutes a null test of non-length geometry. The parity parameter appears as a
rotation of the birefringence axes only for the massive non-length waves.



\section{Conventions}

We work on a flat background in Lorentzian signature with the mostly-plus metric
\begin{equation}
(\eta_{\mu\nu}) = \mathrm{diag}(-1,+1,+1,+1),
\qquad x^0=t, \qquad
\Box \equiv \eta^{\mu\nu}\del_\mu\del_\nu = -\del_t^2+\nabla^2 .
\end{equation}
We denote space-time indices by greek letters $\mu,\nu,\ldots = 0,\ldots,3$, which are raised/lowered with $\eta$. For the restriction to spatial indices we use latin letters
$a,b,\ldots = 1,2,3$, which are raised and lowered with $\delta^{ab}$ and $\delta_{ab}$. The Levi-Civita
symbol is fixed by
\begin{equation}
\eps^{0123}=+1 \qquad\Longrightarrow\qquad \eps_{0123}=-1 ,
\end{equation}
and $\eps_{abc}$ denotes the three-dimensional symbol with $\eps_{123}=+1$.

The (anti-)self-dual Plebanski two-forms evaluated on a Minkowski background are
\begin{equation}\label{eq:Sigmadef}
\Sigma^{\pm i}{}_{\mu\nu}
= \pm \imath\left(\delta^0_\mu\delta^i_\nu-\delta^0_\nu\delta^i_\mu\right)
+ \eps^{i}{}_{jk}\,\delta^j_\mu\delta^k_\nu ,
\end{equation}
Here $i,j,\ldots =1,2,3$ are internal indices, which are raised and lowered with $\delta^{ij}$ and $\delta_{ij}$. We will identify spatial and internal indices.


The Plebanski two-forms satisfy the Lorentzian self-duality condition and the completeness relation
\begin{equation}\label{eq:selfdual}
\tfrac12\,\eps_{\mu\nu}{}^{\lambda\tau}\,\Sigma^{\pm a}{}_{\lambda\tau}
= \pm \imath\,\Sigma^{\pm a}{}_{\mu\nu},
\qquad
\sum_c \Sigma^{\pm }{}_{\!c\, \mu\nu}\Sigma^{\pm c}{}_{\rho\sigma}
= \eta_{\mu\rho}\eta_{\nu\sigma}-\eta_{\mu\sigma}\eta_{\nu\rho}
\mp \imath \,\eps_{\mu\nu\rho\sigma},
\end{equation}
and the reality property $(\Sigma^{+a}{}_{\mu\nu})^* = \Sigma^{-a}{}_{\mu\nu}$.

We will also use the symbols 
\ba
\Box=-\partial_t^2 +\nabla^2 \,, \q \nabla^2=\del^a\del_a  \, .
\ea

\section{Area metric dynamics from modified Plebanski gravity and spin foams}\label{Sec:SF}

Here we will review evidence that area metrics do describe spin foam dynamics \cite{PerezLR}. This can be argued on the microscopic \cite{Dittrich:2023ava} level, from the perturbative continuum limit of effective spin foams \cite{Dittrich:2021kzs,Dittrich:2022yoo,AsanteToAppear}, and from modified Plebanski gravity \cite{Borissova:2022clg}, which underlies spin foam construction. 

To start with consider the (non-chiral) Plebanski action \cite{Plebanski:1977zz} with a Holst term \cite{Holst:1995pc}
\ba\label{PA1}
S&=& \int     \eta_{IJKL}   B^{IJ} \wedge F^{KL}(\omega) \,+\, \frac{1}{2\gamma} \epsilon_{IJKL} B^{IJ} \wedge F^{KL}(\omega)\,-\,\tfrac{1}{2} \phi_{IJKL}B^{IJ}\wedge B^{KL} \q ,
\ea
where we have an $\text{so}(3,1)$-valued two-form $B^{IJ}_{\mu\nu}$ and an $\text{SO}(3,1)$ connection $\omega^{IJ}_\mu$ with curvature $F^{KL}(\omega)$ on a four-dimensional manifold. Here $I,J=0,1,\ldots,3$ and $(IJ)$ are anti-symmetric index pairs for the six dimensional Lie algebra.  There are two invariant bi-linear forms on $\text{so}(3,1)$, namely $\eta_{IJKL}:=\tfrac{1}{2} (\eta_{IK}\eta_{JL}-\eta_{IL}\eta_{JK})$, with $\eta_{IJ}$ the Minkowski metric and $\tfrac{1}{2} \epsilon^{IJKL}$, the Levi-Civita symbol.  $\gamma$ is the Barbero-Immirzi parameter, its inverse multiplies the so-called Holst term. 
 The last term imposes the (primary) simplicity constraints on $B$, via Lagrange multiplier fields $\phi_{IJKL}$, for which we impose the symmetries $\phi_{IJKL}=-\phi_{JIKL}=-\phi_{IJLK}=\phi_{KLIJ}$ and $\phi_{IJKL}\epsilon^{IJKL}=0$. Note that these are again the same algebraic symmetries as for the Riemann tensor, we therefore have 20 independent Lagrange multipliers and thus 20 constraints on $B^{IJ}$. The $B^{IJ}$ two-forms have 36 independent components, but (\ref{PA1}) is invariant under 
$\text{SO}(3,1)$ gauge transformations, leaving 30 gauge invariant quantities.  Imposing the 20 simplicity constraints leaves the 10 degrees of freedom for a length metric. 

The first two terms in the action (\ref{PA1}) define topological BF theory. This theory has an additional gauge symmetry, which amounts to a shift-symmetry for the $B$-fields, removing all propagating degrees of freedom.  (The shift symmetry condition for the area metric couplings ensures that a remnant of this symmetry holds for the kinematical part of the action.) The simplicity constraints break this symmetry, and are therefore essential for turning  topological BF theory into gravity with propagating degrees of freedom. 

The modified Plebanski formalism introduced by Krasnov \cite{Krasnov:2006du,Krasnov:2008fm,Krasnov:2009iy,Speziale:2010cf} replaces the (primary) simplicity constraints with potential terms for the $B$-fields, and in particular mass terms. Application to a chiral formulation of Plebanski (where one works only with the self-dual part of the $B$-fields) revealed that the resulting modification to gravity amounts to a mild deformation of general relativity, with just the massless graviton propagating, and no extra poles \cite{Freidel:2008ku,Krasnov:2009ik}. The mechanism is the same as the one which prevents an extra pole for the area metric action with shift-symmetric couplings, which we will encounter later.  However, for non-chiral Plebanski, replacing all the simplicity constraints with mass terms, one finds eight propagating degrees of freedom (a massless graviton, a massive spin two field, and a scalar), with an unstable behaviour at least on flat backgrounds \cite{Alexandrov:2008fs,Speziale:2010cf}.

Motivated by the appearance of area metric dynamics in the perturbative continuum limit of spin foams \cite{Dittrich:2021kzs,Dittrich:2022yoo} (see below), the work \cite{Borissova:2022clg} only replaces 10 of the simplicity constraints by a potential and keeps the remaining 10 as exact constraints. These 10 constraints are chosen such that the remaining 20 degrees of freedom of the $B$-fields can be used to define an area metric. Integrating out the connection one can thus construct an action for area metrics.  This has been done for the linearized action in \cite{Borissova:2022clg} and resulted in the linearized area metric action with shift symmetric couplings, displayed in (\ref{eq:action}) below.~\footnote{In (\ref{eq:action}) the fields $\chi^\pm$ are normalized as components of the area metric, see (\ref{eq:Pdef}). With this normalization the modified Plebanski theory gives the chiral weights $\beta_\pm=\gamma_\pm=1\pm \imath \gamma^{-1}$ in front of both the coupling to $h_{\mu\nu}$ and the kinetic term of $\chi^\pm$.  A canonical normalization of the kinetic terms would require the complex rescaling $\chi^\pm\to\chi^\pm/\sqrt{\beta_\pm}$, which does not preserve the geometric meaning of $\chi^\pm$ as area metric components and in particular changes their coupling to matter, see Section \ref{sec:EM}.}
  Integrating out the non-length degrees of freedom from the area metric, \cite{Borissova:2022clg} found the same mild deformation for linearized gravity, that is no extra poles, as for chiral version of modified Plebanski theory. 

In contrast, the work \cite{BDK} (see also \cite{Alex:2019ari}) started from area metric actions directly, and constructed the most general linearized actions satisfying (linearized) diffeomorphism symmetry. This did however only lead to a lifting of conditions on coupling constants already present in  (\ref{eq:LEH}) below: the shift symmetry condition on $\beta_\pm$ does not need to be satisfied, and the masses for the non-length fields $\chi^\pm$  (which are $m$ for both $\pm$ sectors in (\ref{eq:action})) can be different.

The replacement of ``sharp" simplicity constraints by ``weaker" simplicity constraints can be motivated from loop quantum gravity  \cite{Rovellibook,Thiemannbook} and spin foams \cite{PerezLR,Hrytseniak:2026vpt}.  There, the $B$-fields are quantized (after applying time gauge which trivializes the boost sector) as angular momentum operators. These do not commute, which leads to a $\gamma$-dependent anomaly in the commutator algebra of the (primary) simplicity constraints \cite{Dittrich:2008ar,Dittrich:2010ey,Dittrich:2012rj,PerezLR}. Thus part of the simplicity constraints cannot be imposed sharply, and are instead imposed in a weaker form \cite{Engle:2007wy,Freidel:2007py, EffSFE,EffSFL}.  This can be also connected to a well-known feature of loop quantum gravity, namely the discreteness of area eigenvalues \cite{Rovelli:1994ge, Ashtekar:1996eg}. Areas are quantized as Casimir of the angular momentum algebra, and are hence discrete; the Barbero-Immirzi parameter appears as a global factor in the spectrum for spatial areas. A loop quantum spin network state can be characterized by the area eigenvalues its assigns to a (possibly generalized) triangulation dual to the network. A consistent length assignment would impose constraints between the areas. These  constraints amount to diophantine equations and admit too few solutions for semi-classical states \cite{EffSFE}.  The configurations described by loop quantum gravity constitute therefore a generalizations \cite{Dittrich:2008va,Dittrich:2008ar} of discretized length geometries, such as length Regge calculus \cite{Regge:1961px}.  Interestingly, these generalized discrete configurations can be mapped to area metrics \cite{Dittrich:2023ava}, justifying area metrics from a microscopic viewpoint.

The fact that the simplicity constraints cannot be imposed sharply, led to worries that spin foams might not lead to gravitons in a semi-classical limit \cite{Bonzom:2009hw,Hellmann:2012kz}. After all, the simplicity constraints turn the topological BF action into one describing gravity. This worry remained unresolved for over a decade and effective spin foams \cite{EffSFE,EffSFL} were introduced to allow progress. Replacing the gauge formulation of spin foams by a higher gauge formulation \cite{Asante:2019lki},  effective spin foams encode dynamics in a much more transparent way, keeping key principles of spin foams, e.g. discrete area eigenvalues and a weak implementation of part of the simplicity constraints.  This allowed on the one hand for an explicit numerical computation of expectations values for the discrete dynamics, which are consistent with curvature excitations \cite{Asante:2020iwm}. On the other hand, it allowed for a perturbative continuum limit on a regular lattice \cite{Dittrich:2021kzs,Dittrich:2022yoo,AsanteToAppear}. The analysis in \cite{Dittrich:2021kzs,Dittrich:2022yoo,AsanteToAppear} showed that at leading order (in long wave lengths, or equivalently vanishing lattice constant) one does obtain the linearized Einstein-Hilbert action, which describes massless gravitons. There is a Weyl squared correction at the next-to-leading order \cite{Dittrich:2022yoo,AsanteToAppear}, which can be explained by area metric variables introduced as collective variables on the lattice, with the length induced part being massless and the non-length part coming with a mass on the scale of the lattice constant. The resulting action for the area metric variables is consistent with the action derived from the modified Plebanski formalism in \cite{Borissova:2022clg}.

\section{Action }\label{sec:action}

\subsection{Field content, action and reality conditions}

An area metric on a four-dimensional manifold is a rank-four covariant tensor with the
symmetries $G_{\mu\nu\rho\sigma}=-G_{\nu\mu\rho\sigma}=G_{\rho\sigma\mu\nu}$, i.e.\ a
metric on the space of bivectors: it assigns areas to infinitesimal parallelograms and
dihedral angles between them, as a length metric does for lengths of and angles between
vectors. Here we restrict to cyclic area metrics, $G_{\mu[\nu\rho\sigma]}=0$, this allows to reconstruct the area metric from measuring areas only \cite{Borissova:2022clg}.  With this requirement the area metric tensors
carry the same algebraic symmetries as the Riemann curvature tensor and has thus the same number of independent components, namely $20$.  

Every length metric induces an area metric
$G^{\rm ind}_{\mu\nu\rho\sigma}(g)=g_{\mu\rho}g_{\nu\sigma}-g_{\mu\sigma}g_{\nu\rho}$,
but not conversely: in four dimensions area metrics are a genuine  enlargement of the
configuration space of gravity, arising e.g.\ in the perturbative continuum limit of spin foams \cite{Dittrich:2022yoo,AsanteToAppear}, 
from modified Plebanski theory \cite{Borissova:2022clg} and are also proposed as field content to describe  the effective dynamics of string theory \cite{Schuller:2005ru,Schuller:2005yt, Ho:2015cza, Borissova:2024cpx} and in the context of AdS/CFT holography \cite{Bhattacharya:2025hag}.

We will consider linearized area metric dynamics on a 
 Minkowski-induced background.  We therefore expand
\begin{equation}\label{eq:pertexp}
G_{\mu\nu\rho\sigma}
=\eta_{\mu\rho}\eta_{\nu\sigma}-\eta_{\mu\sigma}\eta_{\nu\rho}+a_{\mu\nu\rho\sigma},
\end{equation}
and decompose $a$ into its Lorentz-irreducible pieces, in complete analogy with the
Riemann tensor: the trace parts assemble into a symmetric tensor $h_{\mu\nu}$ ($10$
components) --- the length-metric perturbation --- while the trace-free, Weyl-like parts
$w^{\pm}_{\mu\nu\rho\sigma}$, self-dual and anti-self-dual with respect to the
background in the sense of \eqref{eq:selfdual} above, carry the remaining $10$. 
The
latter can be encoded in two symmetric trace-free matrices of spacetime scalars
$\chi^{\pm}_{ab}$ ($a,b=1,2,3$ internal indices) via
$w^{\pm}_{\mu\nu\rho\sigma}=2P^{\pm ab}_{\mu\nu\rho\sigma}\chi^{\pm}_{ab}$, with
projectors $P^{\pm}$ built from the two-forms \eqref{eq:Sigmadef}:
\begin{equation}\label{eq:Pdef}
P^{\pm ab}_{\mu\nu\rho\sigma}
=\frac18\Big(\Sigma^{\pm a}_{\mu\nu}\Sigma^{\pm b}_{\rho\sigma}
+\Sigma^{\pm b}_{\mu\nu}\Sigma^{\pm a}_{\rho\sigma}\Big)
-\frac1{12}\,\delta^{ab}\,\Sigma^{\pm c}_{\mu\nu}\Sigma^{\pm c}_{\rho\sigma},
\qquad
\chi^{\pm}_{ab}=\tfrac12\,P^{\pm}{}_{ab}{}^{\mu\nu\rho\sigma}\,a_{\mu\nu\rho\sigma},
\end{equation}
normalized (via $\Sigma^{\pm a}_{\mu\nu}\Sigma^{\pm b\,\mu\nu}=4\delta^{ab}$) so that
$P^{\pm ab}_{\mu\nu\rho\sigma}P^{\pm cd\,\mu\nu\rho\sigma}
=\delta^{a(c}\delta^{d)b}-\tfrac13\delta^{ab}\delta^{cd}$ is the identity on symmetric
trace-free internal pairs. For trace-free symmetric $\chi^\pm_{ab}$ the first relation
simplifies to
$w^{\pm}_{\mu\nu\rho\sigma}=\tfrac12\,\Sigma^{\pm a}_{\mu\nu}\Sigma^{\pm b}_{\rho\sigma}\chi^{\pm}_{ab}$.

In Lorentzian
signature the self-duality eigenvalues are $\pm i$, so the $\chi^{\pm}$ are complex. We therefore impose as reality condition on the $\chi$-fields that they have to be
conjugate to one another, see Eq.~\eqref{eq:reality} below. This is automatically satisfied, if the $\chi$-fields are constructed from real area metric perturbations via (\ref{eq:Pdef}).

\cite{BDK}  derived the most general quadratic action (modulo absorption of coupling constants via field redefinitions), which is invariant under (linearized) diffeomorphisms.  Following \cite{BDK} we absorb the coupling constant in front of the Einstein-Hilbert term by a rescaling of $h_{\mu\nu}$. In contrast to \cite{BDK} we do \emph{not} rescale the non-length fields $\chi^\pm$, which remain normalized as components of the area metric via (\ref{eq:Pdef}). The coupling constants $\beta_\pm$ then multiply both the coupling of $\chi^\pm$ to the length metric and the kinetic term of $\chi^\pm$, as it is the case for the action derived from modified Plebanski theory \cite{Borissova:2022clg}. We restrict to the case where the masses of the $\chi^+$ and $\chi^-$ fields are the same, i.e.\ where the mass terms carry the same complex weight $\beta_\pm$ as the kinetic terms. We will thus work with the  action $S=\int d^4x\,\Lag$ where 
\begin{equation}\label{eq:action}
\boxed{\;
\Lag \;=\; \Lag_{\rm EH}
+\frac14\sum_\pm\beta_\pm\Big[\,2\,
\Sigma^{\pm a\,\mu\nu}\Sigma^{\pm b\,\rho\sigma}\,
h_{\mu\rho}\,\del_\nu\del_\sigma\chi^\pm_{ab}
-(\del_\mu\chi^\pm_{ab})^2 - m^2(\chi^\pm_{ab})^2\Big]\;}
\end{equation}
with
\begin{equation}\label{eq:LEH}
\Lag_{\rm EH}
= -\tfrac12(\del_\rho h_{\mu\nu})^2
+ (\del^\nu h_{\mu\nu})^2
- (\del^\mu h_{\mu\nu})\del^\nu h^\rho_\rho
+ \tfrac12(\del_\mu h^\rho_\rho)^2 \, .
\end{equation}
Reality of the action for real $h_{\mu\nu}$ requires
\begin{equation}\label{eq:reality}
(\chi^{+}_{ab})^{*} = \chi^{-}_{ab},
\qquad
\beta_\pm = \beta_1 \pm \imath\,\beta_2,\quad \beta_{1,2}\in\mathbb{R},
\end{equation}
so that the $+$ and $-$ terms in \eqref{eq:action} are complex conjugates of one another.

Introducing real fields $\chi_{1\,ab}=\mathrm{Re}\,\chi^+_{ab}$,
$\chi_{2\,ab}=\mathrm{Im}\,\chi^+_{ab}$ we see that the kinetic and mass terms read
$-\tfrac12\big[\beta_1\big((\del\chi_1)^2-(\del\chi_2)^2\big)-2\beta_2\,\del\chi_1\!\cdot\!\del\chi_2\big]
- \tfrac{m^2}{2}\big[\beta_1(\chi_1^2-\chi_2^2)-2\beta_2\,\chi_1\chi_2\big]$. The
quadratic form $\big(\begin{smallmatrix}\beta_1&-\beta_2\\-\beta_2&-\beta_1\end{smallmatrix}\big)$
has the eigenvalues $\pm|\beta_+|$. Thus,
 half of the non-length sector has  kinetic and mass
terms with the `wrong sign'. This appears only in Lorentzian signature (and not for Euclidean signature) \cite{BDK}, where the $\chi^\pm$ fields are complex.

If the couplings satisfy the  \emph{shift-symmetry  condition}
\begin{equation}\label{eq:shiftcond}
\tfrac12\left(\beta_++\beta_-\right)=
\beta_1 = 1 
\qquad\Longleftrightarrow\qquad
\beta_\pm=1\pm\imath\sinh(2\xi) \, ,\q  \xi \in \mathbb{R}
\end{equation}
the \emph{massless} part of the action (\ref{eq:action}) features a gauge symmetry \cite{Borissova:2022clg,BDK}
\ba\label{eq:shiftsym}
\delta h_{\mu\nu} = \Box\,\zeta_{\mu\nu},\qquad
\delta\chi^{\pm ab} = -\,\Sigma^{\pm a\,\mu\rho}\Sigma^{\pm b\,\nu\sigma}\del_\rho\del_\sigma\,\zeta_{\mu\nu} \, .
\ea
where we impose TT gauge  conditions on the symmetric tensor $\zeta_{\mu\nu}$, that is $\zeta_{0\mu}=0$, $\del^a\zeta_{ab}=0$, $\zeta^a_a=0$, so that we have 5 independent gauge parameters.

The mass terms break the symmetry while leaving the kinetic term degenerate. We will see that this has profound repercussions for the dynamics --- we will avoid a ghost pole and instabilities in the (linearized) area metric dynamics, which one would a priori expect because of the indefinite signature of the kinetic and mass terms in the action.

\section{Equations of motion and solutions }\label{sec:eom}

\subsection{Equations of motions}

To obtain the equations of motion we vary the action (\ref{eq:action}) with respect to $h_{\mu\nu}$ and $\chi^{\pm}_{ab}$.

Variation of the Einstein-Hilbert part in the action (\ref{eq:action}) leads to the linearized Einstein tensor $G^{\rm lin\,\mu\nu}$:
\begin{equation}
\frac{\delta}{\delta h_{\mu\nu}}\!\int\!\Lag_{\rm EH}
= \Box h^{\mu\nu} - \del^\mu\del_\rho h^{\rho\nu} - \del^\nu\del_\rho h^{\rho\mu}
+ \del^\mu\del^\nu h + \eta^{\mu\nu}\!\left(\del_\rho\del_\sigma h^{\rho\sigma}-\Box h\right)
\;=\; -2\,G^{\rm lin\,\mu\nu},
\end{equation}

Together with the variation of the mixed term we thus obtain
\begin{equation}\label{eq:heq}
\boxed{\;
G^{\rm lin}_{\mu\nu}
= \frac14\sum_\pm \beta_\pm\,
\mathcal{K}_{\pm\mu\nu}^{ab}
\chi^\pm_{ab}\,  \q \text{with} \q \q  \mathcal{K}_{\pm}^{ab,\mu\nu}=\Sigma^{\pm a\,\mu\rho}\Sigma^{\pm b\,\nu\sigma}\del_\rho\del_\sigma \, .
\;}
\end{equation}

For the variation with respect to $\chi^\pm_{ab}$, remember that these tensors are symmetric and traceless, hence only the symmetric-traceless projection of the variation must vanish. After dividing by the overall factor $\beta_\pm$ we obtain 
\begin{equation}\label{eq:chieq}
\boxed{\;
\big(\Box - m^2\big)\chi^{\pm ab}
+ \Big(\mathcal{K}_\pm^{ab,\mu\nu} -\tfrac13\,\delta^{ab}\,\delta_{ed}  \mathcal{K}_\pm^{ed,\mu\nu} \Big)h_{\mu\nu} = 0 \, . \;}
\end{equation}
The trace can be evaluated with the completeness relation \eqref{eq:selfdual} leading to
\begin{equation}\label{eq:Strace}
\delta_{ed}  \mathcal{K}_\pm^{ed,\mu\nu} h_{\mu\nu} = \Box h^\rho_\rho - \del^\mu\del^\nu h_{\mu\nu},
\end{equation}
which vanishes once transverse--traceless gauge is imposed below.

The equations of motions \eqref{eq:heq} and \eqref{eq:chieq} both contain $\mathcal{K}^{ab,\mu\nu}$; in \eqref{eq:chieq} the internal indices $ab$ are free, in \eqref{eq:heq} the spacetime
indices $\mu\nu$  are free. As we will see, this innocuous-looking difference flips the
sign of the parity-odd part of this operator when applied to symmetric, transverse and traceless tensors. 


\subsection{Reduction to TT gauge}

The equations of motion will become  more transparent after we subject the metric perturbations to transverse--traceless (TT) gauge:
\begin{equation}\label{eq:TTgauge}
h_{0\mu}=0,\qquad \del_a h_{ab}=0,\qquad h_{aa}=0
\qquad\Longrightarrow\qquad h=0,\;\; \del^\mu h_{\mu\nu}=0 .
\end{equation}
In the following TT tensors will mean symmetric, traverse and traceless tensor fields.

For symmetric spatial tensors we define the \emph{tensor curl}
\begin{equation}\label{eq:curl}
(\Curl X)_{ab} \equiv {\eps_{a}}^{cd}\,\del_c X_{db} ,
\end{equation}
which for transverse--traceless $X_{ab}$ is again symmetric, transverse and
traceless.\footnote{In three dimensions the antisymmetric part of $\Curl X$ is dual to
the vector $v_c = \eps_{cab}(\Curl X)_{ab} = \eps_{abc}\,\eps_{ajk}\,\del_j X_{kb}
= \del_b X_{cb} - \del_c X_{bb}$, using
$\eps_{abc}\eps_{ajk}=\delta_{bj}\delta_{ck}-\delta_{bk}\delta_{cj}$. Thus $\Curl$ maps the TT
sector to itself.} We have on TT tensors 
$\Curl^2 = -\nabla^2$.

To simplify the contractions involving ${\cal K}$ we need
\begin{lemma}[${\cal K}$-contractions in TT gauge]\label{lem:contract}
Let $X_{ab}$ be symmetric, transverse and traceless. For $
\mathcal{K}_{\pm}^{ab,\mu\nu}=\Sigma^{\pm a\,\mu\rho}\Sigma^{\pm b\,\nu\sigma}\del_\rho\del_\sigma$, 
using
$\Sigma^{\pm a\, c\nu}\del_\nu = \pm \imath\,\delta^{ac}\,\del_0 + \eps^{acd}\,\del_d$  one has:
\begin{align}
\text{(i) internal indices free:}\qquad
&{\cal K}_\pm^{ab,cd} X_{cd}
= -\Big[\ddot X + \nabla^2 X \pm 2i\,\Curl\dot X\Big]^{ab}\, ,
\label{eq:lemi}\\[1mm]
\text{(ii) spatial spacetime indices free:}\qquad
&{\cal K}_\pm^{ab,cd} X_{ab}
= -\Big[\ddot X + \nabla^2 X \mp 2i\,\Curl\dot X\Big]^{cd}.
\label{eq:lemii}
\end{align}
\end{lemma}

\begin{proof}
Expand $( \pm \imath\,\delta^{ac}\,\del_0 + \eps^{acg}\,\del_g)(\pm \imath\,\delta^{bd}\,\del_0 + \eps^{bdf}\,\del_f)$ 
acting on $X_{cd}$ (case $(i)$) or $X_{ab}$ (case $(ii)$).
The $\del_0^2$ terms give $-\ddot X$ in both cases. The double-$\eps$ terms give, after
using the determinant identity for $\eps^{acg}\eps^{bdf}$ together with transversality and
tracelessness of $X$, the result $-\nabla^2 X$ in both cases. 

The cross terms differ:
in case $(i)$ the $\eps$ of each $\Sigma$ is contracted as follows $\eps^{acd}\del_d X_{c}^b = -(\Curl X)^{ab}$,
producing $\mp 2i\,\Curl\dot X$; in case $(ii)$ the contraction is
$\eps^{acd}\del_d X_{a}^b = +(\Curl X)^{cb}$, producing  $\pm 2i\,\Curl\dot X$ (note that $X$ is symmetric and TT and hence $\Curl X$ is also symmetric).
\end{proof}

~\\

We  therefore introduce the chiral wave operators
\begin{equation}\label{eq:Ddef}
\D^{\pm} \equiv \del_0^2 + \nabla^2 \pm 2\imath\,\Curl\,\del_0 .
\end{equation}
Note the relative $+$ between $\del_t^2$ and $\nabla^2$.  On TT fields,
using $\Curl^2=-\nabla^2$ we have:
\begin{equation}\label{eq:DDbox}
\boxed{\;\D^{\mp}\D^{\pm}
= (\del_t^2+\nabla^2)^2 + 4\,\Curl^2\del_t^2
= (\del_t^2-\nabla^2)^2 = \Box^2\;}
\end{equation}

With the chiral wave operates we can express the equations of motions in TT gauge as follows. 

The spatial components of the $h$-equation (\ref{eq:heq}) reduce in TT gauge to
\begin{equation}\label{eq:hTT}
\boxed{\;
\Box h_{ab} = \frac12\sum_\pm \beta_\pm\,(\D^{\mp}\chi^{\pm})_{ab} \, .
\:}
\end{equation}

For the $(0\mu)$ components of the $h$-equation (\ref{eq:heq}) we notice that with TT gauge
$G^{\rm lin}_{00}=G^{\rm lin}_{0a}=0$. So the corresponding components of \eqref{eq:heq} turn into
conditions for the longitudinal modes of $\chi^\pm$:
\begin{equation}\label{eq:constraints}
\sum_\pm\beta_\pm\,\del^a\del^b\chi^\pm_{ab}=0,
\qquad
\sum_\pm\beta_\pm\Big[\del^b\dot\chi^\pm_{ba}
\mp \imath\,{\eps_{a}}^{cd}\del_c\big(\del^b\chi^\pm_{bd}\big)\Big]=0 .
\end{equation}

The $\chi$-equation (\ref{eq:chieq}) reads now as
\begin{equation}\label{eq:chiTT}
\boxed{\;
\big(\Box-m^2\big)\chi^\pm_{ab} = (\D^{\pm}h)_{ab}.
\;}
\end{equation}
Since $\D^\pm$ preserves the TT property we can restrict to
\emph{transverse} $\chi^\pm_{ab}$ (that is $\del_a\chi^\pm_{ab}=0$), which we do from now on. (Even without imposing TT gauge on the metric perturbations, one can show that  $\D^\pm$ preserves the space of longitudinal modes and thus the three longitudinal-traceless $\chi$-modes decouple
from the TT mode dynamics.) Using TT gauge  equation (\ref{eq:chiTT}) shows that the longitudinal modes are massive. For the remainder of this work we restrict to the case that the longitudinal modes are vanishing, which is consistent with (\ref{eq:constraints}). Thus  $\chi^\pm$ is symmetric, traceless (by definition) and transverse.

Solving the $\chi$-equation by
$\chi^\pm = (\Box-m^2)^{-1}\D^\pm h$ (plus homogeneous solutions, which
constitute the massive branch found below), and inserting into the $h$-equation, the identity
\eqref{eq:DDbox} and the shift symmetry condition $\beta_++\beta_-=2$ gives
\begin{equation}
\Box h = \tfrac12\Big(\beta_++\beta_-\Big)(\Box-m^2)^{-1}\Box^2 h
\quad\Longrightarrow\quad     (\Box-m^2)\Box h = \Box^2 h  \quad\Longrightarrow\quad 
m^2\,\Box\,h_{ab}=0 .
\end{equation}
Thus, the shift symmetry condition $\beta_++\beta_-=2$ is essential in avoiding an extra spin-2 pole, and leads to a ghost-free propagator for the effective action for the metric perturbations, see \cite{Borissova:2022clg,BDK}. 
For generic couplings, with 
$\rho\equiv\tfrac12(\beta_++\beta_-)=\beta_1$, the same manipulation
gives $\big[(1-\rho)\Box-m^2\big]\Box h=0$: an additional pole at
$\Box = m^2/(1-\rho)$, which runs to infinity exactly on the shift-symmetry surface
$\rho=1$.

\subsection{Expansion into modes and solutions}\label{sec:modes}

To solve the equation of motion we use an expansion in Fourier modes $X_{cd}(t,x)=\int d\omega d^3k \,e^{\imath(\vec{k}\cdot \vec{x}-\omega t)} X_{cd}(\omega,\vec{k})$.
We introduce momentum vectors $\hat{\vec{k}}^\theta$ and $\hat{\vec{k}}^\varphi$, such that
$\qty(\frac{\vec{k}}{k}, \hat{\vec{k}}^\theta,\hat{\vec{k}}^\varphi)$, where $k=\sqrt{\vec{k}\cdot\vec{k}}$, defines a  right-handed orthonormal basis, i.e.,
\be\label{eq:kIdentities}
\hat{k}_a^\theta  \hat{k}_b^\theta \delta^{ab} = \hat{k}_a^\varphi \hat{k}_b^\varphi\delta^{ab} =1\,,\quad\quad\quad \hat{k}_a^\theta  \hat{k}_b^\varphi \delta^{ab} = \hat{k}_a^\theta  k_b \delta^{ab} =\hat{k}_a^\varphi k_b \delta^{ab}  =0\,,\quad\quad\quad \epsilon^{abc}\frac{k_{a}}{k}\hat{k}_b^\theta \hat{k}_c^\varphi = +1\,.
\ee
We can then define a helicity basis for the TT modes
\be
t_{ab}^R \equiv  \frac{1}{2}\left( \hat k^\theta_a+\imath \hat k^\varphi_a\right)\left( \hat k^\theta_b+\imath \hat k^\varphi_b\right)
\,,\quad\quad \quad t_{ab}^L \equiv  \frac{1}{2}\left( \hat k^\theta_a-\imath \hat k_a^\varphi\right)\left( \hat k^\theta_b-\imath \hat k_b^\varphi\right)\,.
\ee
An alternative choice is to use cross and plus polarization
\be\label{PlusCross}
t_{ab}^+ \equiv \frac{1}{\sqrt{2}} \qty(\hat{k}_a^\theta \hat{k}_b^\theta - \hat{k}_a^\varphi \hat{k}_b^\varphi )\,,\quad\quad \quad t_{ab}^{\cross} \equiv \frac{1}{\sqrt{2}} \qty(\hat{k}_a^\theta \hat{k}_b^\varphi + \hat{k}_a^\varphi \hat{k}_b^\theta)\, .
\ee
The curl acts on these basis elements as 
\begin{equation}\label{eq:curlaction}
\Curl\, t^{R} = +k\,t^{R},\qquad
\Curl\, t^{L} = -k\,t^{L},
\qquad
\Curl\, t^{+} = +ik\,t^{\times},\qquad
\Curl\, t^{\times} = -ik\,t^{+}.
\end{equation}
Due to the action of the tensor curl, the equations of motions decouple in the helicity basis, but mix plus and cross polarizations.

\subsubsection{Helicity basis \texorpdfstring{($R/L$)}{(R/L)}}\label{sec:helicity}

On a mode
$\propto t^{H}_{ab}e^{i(\vec k\cdot\vec x-\omega t)}$ of helicity
$H=+1$ ($R$) or $H=-1$ ($L$),
\begin{equation}
\Box \;\to\; P \equiv \omega^2-k^2,
\qquad
\D^{\pm}\;\to\; -(\omega^2+k^2)\pm 2H\omega k
= -(\omega\mp H k)^2 \equiv -\,Q^H_\pm ,
\end{equation}
and the equations of motions can now be expressed as
\begin{equation}\label{eq:scalarsys}
(P-m^2)\,\chi^{\pm,H} = -\,Q^H_\pm\,h^H ,
\qquad
P\,h^H = -\frac12\Big(\beta_+\,Q^H_-\,\chi^{+,H} + \beta_-\,Q^H_+\,\chi^{-,H}\Big).
\end{equation}
Note the crosswise pairing $\beta_+ Q_-$, $\beta_- Q_+$ in the $h$-equation ---
which appears due to the chirality flip $\D^\pm\to\D^\mp$ of Lemma~\ref{lem:contract}.

For each tuple $(\omega,\vec{k},H)$  equations (\ref{eq:scalarsys}) give three linear equations for three mode coefficients $h^H, \chi^{+,H}$ and $\chi^{-,H}$.  This system can be described by a $3\times 3$ matrix. For the system to admit solutions the determinant of this matrix has to vanish. Evaluating the determinant by using the two identities
\begin{equation}\label{eq:keyids}
Q^H_+Q^H_- = (\omega-H k)^2(\omega+H k)^2 = (\omega^2-k^2)^2 = P^2,
\qquad
\tfrac12\big(\beta_++\beta_-\big) = \beta_1 = 1,
\end{equation}
we see that the determinant facorizes
\begin{equation}\label{eq:det}
-\,m^2\,P\,(P-m^2) = 0 .
\end{equation}
 Hence, \emph{per helicity},
\begin{equation}\label{eq:spectrum}
\boxed{\;
P=0:\;\; \omega = \pm k
\qquad\text{and}\qquad
P=m^2:\;\; \omega = \pm\sqrt{k^2+m^2}\;}
\end{equation}

Thus, despite the indefinite kinetic and mass terms, we find a stable dynamics. That is four real frequencies per helicity, independent of the parity-violation parameter  $\xi$  (defined via $\beta_\pm=1\pm\imath\sinh(2\xi)$).

In general, as the equations of motion can be encoded into $3\times3$ matrices, one would expect a determinant cubic in $P$ (or equivalently cubic in $\omega^2$). However, due to the shift-symmetry conditions, the cubic terms cancel each other out. This is closely related to the appearance of constraints in the canonical analysis \cite{BDK}.  We can isolate the second order time derivatives from the equations of motion, that is write the equations of motion in the form
\ba
A \,(\ddot h,\ddot\chi^+,\ddot\chi^-)^{T} = b \, ,
\ea
where $b$ contains only zeroth or first order time derivative terms and $A$ is given by
\begin{equation}\label{eq:accmatrix}
A=\begin{pmatrix}1 & \beta_+/2 & \beta_-/2\\[1mm] 1 & 1 & 0\\[1mm] 1 & 0 & 1\end{pmatrix},
\qquad
\det A = 1-\tfrac12\big(\beta_++\beta_-\big) = 1-\beta_1,
\end{equation}
Thus $A$ has a right null vector  $A\,v=0$ for $v=(1,-1,-1)^{T}$, which is the direction of the shift symmetry (\ref{eq:shiftsym}).  The accelerations along this field-space direction are left undetermined by the equations of motion, and the Legendre map from velocities to momenta is not invertible as it is onto a (primary) constraint hypersurface in phase space, which fixes a combination of momenta as function of the fields.

There is also a left null vector $u^{T}A=0$ for
$u=(1,-\beta_+/2,-\beta_-/2)$. This describes a combination of the equations of motion which contains
no accelerations (that is $\omega^2$ terms). It is given by
\ba\label{eq:secondary}
&&\left( - 2k^2+(\beta_+-\beta_-)H\omega k\right) h^H+ \nn\\ 
&&\q\q\q\q\q \left(2 k^2 + m^2\right)\tfrac{1}{2}\left( \beta_+ \chi^{+,H} + \beta_- \chi^{-,H} \right)  + 2H\omega k  \tfrac{1}{2}\left( \beta_+ \chi^{+,H} - \beta_- \chi^{-,H} \right) \,=\,0 \q ,
\ea
and corresponds to the secondary constraint found in the canonical analysis of \cite{BDK}. Here $\chi^H:=\tfrac{1}{2}( \beta_+ \chi^{+,H} + \beta_- \chi^{-,H})$ appears without any  time derivative, and one can use the constraint (\ref{eq:secondary}) to solve for $\chi^H$. The primary constraint described above can be used to remove the momentum conjugated to $\chi^H$, leaving a reduced phase space with configuration variables $h^H$ and $\phi^H:=\tfrac{\imath}{2}( \beta_+ \chi^{+,H} - \beta_- \chi^{-,H})$, as found in \cite{BDK}.

~\\

Returning to the construction of solutions, we continue with describing in more detail the two solution branches.

{\bf Dressed gravitons:} Here we have $P=0$, that is $\omega=\pm k$.   The first equation (\ref{eq:scalarsys}) thus amounts to 
\ba
\chi^{\pm,H} = \frac{Q^H_\pm}{m^2}h^H \q .
\ea
With $Q^H_\pm=(\omega\mp Hk)^2$ we have $Q^H_+=0$ and $Q^H_-=4k^2$  if $\text{sign}(\omega H)=+1$  and  $Q^H_+=4k^2$ and $Q^H_-=0$ if $\text{sign}(\omega H)=-1$. Thus, depending on the sign, either $\chi^+$ or $\chi^-$ vanishes, and one can see that the second equation in (\ref{eq:scalarsys}) is also satisfied. We note that an $(R,+\omega)$ and an $(L,-\omega)$ mode lead to the same sign, which is consistent with the reality condition \eqref{eq:reality}.  Indeed, physical (frame-independent) helicity is encoded by $\text{sign}(\omega H)$. A graviton with positive physical helicity drags along a $\chi$-field with negative chirality and vice versa. Note that the dressing does not depend on the couplings $\beta_\pm$, and in particular not on the parity parameter $\xi$.

{\bf Massive modes:} Here we have $P=m^2$, that is $\omega=\pm\sqrt{k^2+m^2}$.
Both
$Q^H_\pm=(\omega\mp Hk)^2$ are nonzero, so the first equation of \eqref{eq:scalarsys}
forces $h^H = 0$. The second equation then relates $\chi^{+,H}$ and $\chi^{-,H}$
\begin{equation}\label{eq:massivebranch}
\beta_+ Q^H_-\,\chi^{+,H} + \beta_- Q^H_+\,\chi^{-,H} = 0 \, .
\end{equation}
Equation \eqref{eq:massivebranch} is compatible with the reality condition \eqref{eq:reality}, which holds for the fields and not mode by mode. Complex conjugation maps a mode $(k,\omega,H)$ to its partner $(-k,-\omega,-H)$ and exchanges the chiral sectors, $\chi^{\pm}(-k,-\omega,-H)=\overline{\chi^{\mp}(k,\omega,H)}$. Since $Hk$ is unchanged and $\omega\to-\omega$, this maps $Q^H_\pm\to Q^H_\mp$, and together with $\overline{\beta_\pm}=\beta_\mp$ the relation \eqref{eq:massivebranch} for the partner mode is the complex conjugate of the one for the original mode. A real massive wave is the superposition of a mode and its partner, with one free complex amplitude. We thus have one massive degree of freedom per helicity. It is a genuinely non-length excitation which
carries no length-metric perturbation at all. In contrast to the dressing of the graviton, the massive branch does depend on the parity parameter: the relative phase of the self-dual and anti-self-dual components is fixed by $\beta_+/\beta_-=e^{2\imath\arctan(\sinh 2\xi)}$.

\subsubsection{Real polarization basis \texorpdfstring{($+/\cross$)}{(+/x)}}\label{sec:pluscross}

In the real polarization basis the tensor curl \eqref{eq:curlaction} mixes the two polarizations $+$ and $\cross$. The equations of motion for these modes are now described by a $6\times 6$ matrix, with a coupling between the $+$ and $\cross$ modes.  Instead of going through these equations of motion we can construct solutions in the $(+,\cross)$ basis as superpositions of the helicity mode solutions.  To this end we note that the following transformation for the components $(X^+,X^{\cross} )$ and $(X^R,X^L)$ with respect to the two tensor bases applies
\ba
X^+ = \tfrac{1}{\sqrt{2}}(X^R+X^L)\,,\q 
X^{\cross} = \tfrac{\imath}{\sqrt{2}} (X^R-X^L)\,.
\ea
In order to work with real fields, we replace $\chi^\pm=\chi_1\pm \imath \chi_2$ with $(\chi_1,\chi_2)$. 

For a plus polarized gravitational wave with $h^+=A$  and $h^{\cross}=0$ and $\omega=+k$ one then finds 
\begin{equation}\label{eq:chi12dress}
\big(\chi_1^{+},\,\chi_2^{+},\,\chi_1^{\times},\,\chi_2^{\times}\big)
=\frac{2k^2A}{m^2}\,\big(1,\ 0,\ 0,\ -1\big) \, .
\end{equation}
Thus a plus polarized gravitational wave comes with plus and cross polarized $\chi$-excitations.  In general, a length metric perturbation $h_{ab}=At^P_{ab} \exp(\imath(\vec{k}\cdot \vec{x}-\epsilon k t))$ with $P=+,\cross$ and $\epsilon=\pm 1$ comes with the following $\chi$-amplitudes:
\begin{equation}\label{eq:chi12gen}
\big(\chi_1^{+},\,\chi_2^{+},\,\chi_1^{\times},\,\chi_2^{\times}\big)
=\frac{2k^2A}{m^2}\times
\begin{cases}
\big(1,\ 0,\ 0,\ -\epsilon\big), & P=+,\\[1mm]
\big(0,\ \epsilon,\ 1,\ 0\big), & P=\times .
\end{cases}
\end{equation}

\subsection{Remarks on negative energy terms}

The action \eqref{eq:action} featured  the `wrong sign' for half of the 
 kinetic, gradient and mass terms for the $\chi$-fields: the quadratic form of the real fields $(\chi^1,\chi^2)$ has one positive and one negative eigenvalue, $\pm|\beta_+|$.  Let us discuss the fate of these negative energy terms for solutions. For plane waves $\chi^{1,2}_{ab}=\mathrm{Re}\big[X^{1,2}_{ab}\,e^{\imath(kz-\omega t)}\big]$
the time-averaged energy density of the $\chi$-sector is
\be\label{ChiEnergy}
\langle\rho_\chi\rangle
=\tfrac14\big(\omega^2+k^2+m^2\big)\sum_{ab}\Big(\beta_1\big(|X^{1}_{ab}|^{2}-|X^{2}_{ab}|^{2}\big)-2\beta_2\,\mathrm{Re}\big(X^1_{ab}\overline{X^2_{ab}}\big)\Big) \q .
\ee
For the dressed gravitons this vanishes exactly: the dressing
\eqref{eq:chi12gen} comes with the properties  $\sum_{ab}|X^{1}_{ab}|^{2}=\sum_{ab}|X^{2}_{ab}|^2$ and $\sum_{ab}\mathrm{Re}(X^1_{ab}\overline{X^2_{ab}})=0$. The vanishing of (\ref{ChiEnergy}) extends to arbitrary superpositions (due to orthogonality of the polarization basis)
and the energy of a dressed graviton is only corrected by the $h$--$\chi$ cross couplings at relative order
$(k/m)^{2}$. 

The massive modes behave differently:
 there $h_{ab}\equiv0$, the
energy is purely of $\chi$-type.  With \eqref{eq:massivebranch} we have $X^-=-\frac{\beta_+ Q_-^H}{\beta_- Q_+^H}X^+$ for positive and negative helicity $H$ and thus
\ba
\langle\rho_\chi\rangle
=- \tfrac14\big(\omega^2+k^2+m^2\big) \frac{Q^H_-}{Q^H_+}\, \beta_1\sum_{ab}|X^+_{ab}|^2
\ea
with $\beta_1=1$ and $Q^H_\pm$ being positive, so that these excitations carry
strictly \emph{negative} energy, for both helicities.  

We contrast these findings to \cite{Boulware:1983td}, 
which considered purely quadratic, scale-invariant gravity. The work \cite{Boulware:1983td} shows that an exact zero-energy theorem holds for isolated systems. For area metric actions scale invariance is broken by the Einstein-Hilbert term and by the mass $m$ for the non-length degrees of freedom. The graviton carries (modulo a $(k/m)^{2}$ correction) the standard positive energy, and the indefiniteness is confined to the gapped (Planck-) massive branch.

Importantly, at the linearized level the
negative-energy branch does not affect the stability of the dynamics. The negative energy terms are even tied to the mechanism that protects the graviton. For the interacting level it raises the usual concerns. For a Planckian mass we are outside the regime of effective descriptions, and the question passes to the ultraviolet completion.

\subsection{Violating the shift symmetry condition}

Let us consider a relaxation of the shift symmetry condition \eqref{eq:shiftcond}, that is
$\rho\equiv\beta_1\neq1$. 
The equations of motion still decouple in the helicity basis and the determinant of the linear equation system
\eqref{eq:scalarsys} is
\ba\label{eq:detrho}
\det \;\propto\; -P(P-m^2)((1-\rho)P-m^2)
\ea
The massless graviton branch ($P=0$) and the massive $\chi$ branch ($P=m^2$) are therefore not changed: the graviton remains exactly massless for any $\rho$, with the
identical dressing $\chi^\pm=Q^H_\pm h/m^2$, and the massive branch remains at $P=m^2$ with
$h_{ab}\equiv0$. Also, the vanishing of the $\chi$-energy of the dressing in the graviton branch
persists.  

Thus, the main effect of violating the shift-symmetry condition is the appearance of a third root $P_*=m^2/(1-\rho)$.
Unlike the  massive $\chi$ branch at $P=m^2$, this branch does carry  a non-vanishing length-metric perturbation. Integrating out the $\chi$-perturbations we can find the propagator $G_{\text{eff}}$ for the effective length metric action (see also \cite{BDK}), to be given as
\ba\label{Effhprop}
G_{\text{eff}}&=& \frac{1}{P}+\frac{\rho/(1-\rho)}{P-m^2/(1-\rho)}  \, ,
\ea 
 with the massless graviton pole and the pole at $P=P_*=m^2/(1-\rho)$. (The massive $\chi$ branch does not involve length metric perturbations and therefore does not appear in the effective propagator.)
 
Assuming $m\sim M_{\text{Planck}}$, this leads to different regimes:
\begin{itemize}
\item For $\rho>1$ the extra branch is tachyonic,
$\omega^2=k^2-m^2/(\rho-1)$, with growing modes of maximal rate $m/\sqrt{\rho-1}$.  This instability lies within the validity of the effective theory only for $\rho\gtrsim2$, where it proceeds at a Planckian rate and is excluded by the persistence of a long-lived vacuum
\item   For $0<\rho<1$ the extra branch is a massive
length-metric mode of mass $m/\sqrt{1-\rho}\geq m$, i.e. super-Planckian near
the shift point, whose only low-energy trace is a contact correction of order
$\rho/m^2$ to the graviton propagator --- unobservably close to the shift-symmetric
theory.
\item  For the extra branch to be light we would require $\rho\ll-1$, but (\ref{Effhprop}) shows that its residue is $\rho/(1-\rho)$ and negative. We would thus have a genuine massive length metric spin-2 ghost, largely excluded because it would lead to e.g. additional radiation channels for binaries and Yukawa corrections to the Newtonian potential of range $\sqrt{1-\rho}/m$.
\end{itemize} 
 Assuming $m\sim M_{\text{Planck}}$,  the shift symmetry condition is thus observationally safe on both sides: large violations are ruled
out, small violations are invisible at accessible energies.  We note that the area metric action constructed from modified Plebanski theory and from the continuum limit of effective spin foams satisfy $\rho=1$ by construction rather than tuning.  Conversely, the phenomenology developed in this work  --- exact masslessness of the graviton branch, helicity degeneracy, the dressing and its energetics, and the interferometer polarization channel discussed below --- does not hinge on the
shift symmetry condition.

\section{Electro-magnetic response}\label{sec:EM}

\subsection{Coupling of electro-magnetism to area metrics}

We have seen that the propagation of the length-metric perturbation itself remains mostly unaffected by the presence of the non-length degrees of freedom.  But non-vanishing gravitons lead to non-vanishing non-lengths excitations $\chi^\pm$.  One might be able to probe the presence of such excitations via the electromagnetic field. 

Area metrics couple naturally to electromagnetism, the field strength $F_{\mu\nu}$ is a two-form, and the Maxwell action
requires precisely a metric on bivectors. We therefore extend the action
\eqref{eq:action} by the constitutive (premetric) Maxwell term
\begin{equation}\label{eq:SEM}
S_{\rm EM}[A;G]=-\frac18\int d^4x\;\mathcal G^{\mu\nu\rho\sigma}\,F_{\mu\nu}F_{\rho\sigma}\,,
\qquad F_{\mu\nu}=\del_\mu A_\nu-\del_\nu A_\mu\,,
\end{equation}
where the constitutive tensor $\mathcal G^{\mu\nu\rho\sigma}$ is the density-weighted
\emph{inverse} of the area metric on bivectors, normalized such that for an induced area
metric $G^{\rm ind}(g)$ one has
$\mathcal G^{\mu\nu\rho\sigma}=\sqrt{-g}\,\big(g^{\mu\rho}g^{\nu\sigma}-g^{\mu\sigma}g^{\nu\rho}\big)$,
so that \eqref{eq:SEM} reduces to the Maxwell action
$-\tfrac14\int\sqrt{-g}\,F_{\mu\nu}F^{\mu\nu}$. This is the setting of premetric (local
and linear) electrodynamics \cite{Obukhov:2000nw,Hehl:2004zz} and of constructive gravity \cite{Punzi:2007di,Schneider:2017tuk}.

Expanding to first order in the perturbation \eqref{eq:pertexp}: since $\mathcal G$
inverts $G$ on bivectors, the perturbation enters with a minus sign relative to a
naive raising of indices,
$\mathcal G^{\mu\nu\rho\sigma}
=\eta^{\mu\rho}\eta^{\nu\sigma}-\eta^{\mu\sigma}\eta^{\nu\rho}-a^{\mu\nu\rho\sigma}
+(\text{trace terms})$,
where $a^{\mu\nu\rho\sigma}$ is raised with $\eta$ and the density weight contributes
only trace terms, irrelevant for transverse--traceless $h$ and the trace-free
$\chi^\pm$. Hence
\begin{equation}\label{eq:LEMlin}
\Lag_{\rm EM}=-\frac14\,F_{\mu\nu}F^{\mu\nu}
+\frac18\,a^{\mu\nu\rho\sigma}F_{\mu\nu}F_{\rho\sigma}\, .
\end{equation}
The variation of this action with respect to $A_\mu$ leads to the inhomogeneous Maxwell equation
\begin{equation}\label{eq:Maxmod}
\del_\mu\Big(F^{\mu\nu}-\tfrac12\,a^{\mu\nu\rho\sigma}F_{\rho\sigma}\Big)=0 \q .
\end{equation}
The electromagnetic field acts also as a source for the area metric, this happens however via terms quadratic in the electromagnetic field strength, affecting $\chi$ by an amount $\sim F^2/(M_{\text{Planck}}m)^2$. We will neglect this effect and rather assume a regime where the electromagnetic field acts as a probe for the area metric.

We now evaluate the coupling term $a F F$ in  \eqref{eq:LEMlin} more explicitly. With our conventions we have 
$E_a=F_{a0}$ and $B_a=\tfrac12\eps_{abc}F_{bc}$. The metric perturbation is
\begin{equation}\label{eq:apert}
a_{\mu\nu\rho\sigma}
\;=\;
h_{\mu\rho}\,\eta_{\nu\sigma}
+\eta_{\mu\rho}\,h_{\nu\sigma}
-h_{\mu\sigma}\,\eta_{\nu\rho}
-\eta_{\mu\sigma}\,h_{\nu\rho}\;+\;\sum_\pm w^{\pm}_{\mu\nu\rho\sigma}\,,
\q\q
w^{\pm}_{\mu\nu\rho\sigma}
=\frac12\,\Sigma^{\pm a}{}_{\mu\nu}\,\Sigma^{\pm b}{}_{\rho\sigma}\,\chi^{\pm}_{ab}\,,
\end{equation}
Together with $\chi^\pm=\chi_1\pm\imath\chi_2$ we obtain
\begin{equation}\label{eq:EMcoupling}
\frac18\,a^{\mu\nu\rho\sigma}F_{\mu\nu}F_{\rho\sigma}
= -\,\tfrac12\,h^{ab}\big(E_aE_b+B_aB_b\big)
-\tfrac12\,\chi_{1}^{ab}\big(E_aE_b-B_aB_b\big)
-\tfrac12\,\chi_{2}^{ab}\,\big(E_aB_b+B_aE_b\big)
\end{equation}
The three symmetric trace-free fields $h, \chi_1,\chi_2$  couple to the three independent
quadrupolar quadratic forms of the electromagnetic field: $h$ to the stress
(energy-momentum) quadrupole $E_aE_b+B_aB_b$; $\chi_1$ to the
Lagrangian-type, parity-even quadrupole $E_{a}E_{b}-B_{a}B_{b}$; and $\chi_2$ to the
Pontryagin-type, parity-odd quadrupole $E_{(a}B_{b)}$. In this sense $\chi_1$ and
$\chi_2$ are the spin-2 analogues of a dilaton coupling $\varphi\,F_{\mu\nu}F^{\mu\nu}$
and an axion coupling $\theta\,F_{\mu\nu}\tilde F^{\mu\nu}$, respectively.

\subsection{Optical effects of area metric perturbations}

To discuss the optical effects of $\chi_1$ and $\chi_2$  (which in this section, to improve readability,  we will denote by $\chi^1$ and $\chi^2$) we introduce the excitation fields 
\begin{equation}\label{eq:DH}
D_a\equiv\frac{\del\Lag_{\rm EM}}{\del E_a}
=\sum_b\varepsilon_{ab}E_b-\sum_b \chi^2_{\,ab}B_b\,,
\qquad
H_a\equiv-\frac{\del\Lag_{\rm EM}}{\del B_a}
=\sum_b (\mu^{-1})_{ab}B_b+\sum_b \chi^2_{\,ab}E_b\,,
\end{equation}
where 
\begin{equation}\label{eq:constitutive}
\varepsilon_{ab}=\delta_{ab}-h_{ab}-\chi^1_{\,ab},\qquad
(\mu^{-1})_{ab}=\delta_{ab}+h_{ab}-\chi^1_{\,ab} \, .
\end{equation}
The  Maxwell equations take the familiar form of electrodynamics in
a medium,
\begin{equation}\label{eq:maxwellset}
\underbrace{\;\vec\nabla\cdot\vec B=0\,,\qquad
\vec\nabla\times\vec E=-\del_t\vec B\;}_{\text{homogeneous}}\,,
\qquad\qquad
\underbrace{\;\vec\nabla\cdot\vec D=0\,,\qquad
\vec\nabla\times\vec H=\del_t\vec D\;}_{\text{inhomogeneous}}\,.
\end{equation}
with the homogeneous pair following from the definition $F_{\mu\nu}=\del_\mu A_\nu-\del_\nu A_\mu$ and  the inhomogeneous pair results from the field equation \eqref{eq:Maxmod}.

 Consider a photon with the plane-wave ansatz
$(\vec E,\vec B)(t,\vec{x}) =  (\vec E,\vec B)(\hat n, \omega_\gamma)\times e^{\imath(q\,\hat n\cdot\vec x-\omega_\gamma t)}$, propagating in
the direction $\hat n$, with $\omega_\gamma$ far above the frequencies of the
background fields, which can therefore be taken constant over many optical cycles.
Choose a right-handed orthonormal frame $(\hat n,\hat e_1,\hat e_2)$ and introduce the
circular polarization vectors
\begin{equation}\label{eq:circbasis}
e_\pm=\tfrac{1}{\sqrt2}\big(\hat e_1\pm\imath\,\hat e_2\big)\,,
\qquad
\hat n\times e_\pm=\mp\,\imath\,e_\pm\,,
\qquad
e_+ \cdot e_-=1\,,\qquad e_\pm\cdot e_\pm=0\,.
\end{equation}

Decompose
$\vec E=E_+\,e_+ + E_-\,e_- + E_{\hat n}\,\hat n$ and likewise $\vec B$.  Note that $X_\pm=e_\mp \cdot \vec{X}$.
 For given photon frequency $\omega_\gamma$,
the quantities to be determined are the wavenumber $q$ and the amplitudes of $\vec E$
and $\vec B$ as function of $\hat n$ and $\omega_\gamma$.

With our ansatz, $\vec\nabla\to\imath q\,\hat n$ and $\del_t\to-\imath\omega_\gamma$,
so the four Maxwell equations \eqref{eq:maxwellset}, together with \eqref{eq:DH}, amount to the algebraic conditions
\ba\label{eq:algebraic}
&\hat n\cdot\vec B=0\,,\qquad
&q\,\hat n\times\vec E=\omega_\gamma\,\vec B\,,\nn\\
& \hat n\cdot (\varepsilon\cdot \vec{E} -(\chi^2)\cdot \vec{B}) =0\,,\qquad
&q\,\hat n\times((\mu^{-1})\cdot \vec{B}+(\chi^2)\cdot \vec{E})=-\,\omega_\gamma\,(\varepsilon\cdot \vec{E} -(\chi^2)\cdot \vec{B})\,.
\ea
The second condition determines the magnetic amplitude completely:
\begin{equation}\label{eq:Bcirc}
B_\pm=\mp\,\imath\,\frac{q}{\omega_\gamma}\,E_\pm\,,\qquad B_{\hat n}=0\,,
\end{equation}
with the first condition being redundant. 

The third equation in (\ref{eq:algebraic}) (together with equation (\ref{eq:Bcirc})) fixes the
longitudinal electric amplitude in terms of $E_\pm$:
\ba
(1-(h+(\chi^1))_{\hat n \hat n} )E_{\hat n} = \sum_{S=\pm} (h+\chi^1)_{\hat n S} E_S -\sum_{S=\pm} S\frac{q}{\imath \omega_\gamma} (\chi^2)_{\hat n S} E_S \q .
\ea
We note that $E_{\hat n}$ is at least first order in area metric perturbations. 

The $\hat n$-component of the last condition in (\ref{eq:algebraic}) is satisfied due to the third condition. For the $\pm$ components of the last condition we note that $E_{\hat n}$, which is of first order, enters always multiplied by another area metric perturbation. We will neglect the resulting second order terms and, after using (\ref{eq:Bcirc}), obtain two equations only involving $E_\pm$. A straightforward calculation gives\footnote{Our definition for the components of a $2\times 2$ matrix $X$ is $X\cdot e_+=X_{-+}e_++X_{++}e_-$ and $X\cdot e_-=X_{--}e_++X_{+-}e_-$.}
\ba
{\scriptstyle \left[   
(\omega_\gamma^2-q^2)(1-\chi^1_{-+})\!-\!(\omega_\gamma^2+q^2)h_{-+} 
  \right]}E_+&+&
  {\scriptstyle\left[-(\omega_\gamma^2-q^2)h_{--}\!-\!(\omega_\gamma^2+q^2)\chi^1_{--}\!  -2\imath q\omega_\gamma \chi^2_{--}   \right]}E_- \,\simeq\,0 \, , \nn\\
{\scriptstyle \left[
-(\omega_\gamma^2-q^2)h_{++}\!-\!(\omega_\gamma^2+q^2)\chi^1_{++}
+2\imath q\omega_\gamma \chi^2_{++} 
\right]}E_+&+&
{\scriptstyle \left[
(\omega_\gamma^2-q^2)(1-\chi^1_{+-})\!-\!(\omega_\gamma^2+q^2)h_{+-} 
\right]}E_ -  \,\simeq\,0\, .
\ea
where $\simeq$ denotes equality up to second order term.  We now utilize that $q = \omega_\gamma$ modulo first order terms. This leads to the following simplification
\ba
{\scriptstyle \left[   
(\omega_\gamma^2-q^2)- 2\omega_\gamma^2 h_{-+} 
  \right]}E_+  &+&
  {\scriptstyle\left[-\! 2\omega_\gamma^2\chi^1_{--}\!  -2\imath \omega_\gamma^2 \chi^2_{--}   \right]}E_- \,\simeq\,0 \, , \nn\\
{\scriptstyle \left[
-2\omega_\gamma^2\chi^1_{++}
+2\imath \omega_\gamma^2 \chi^2_{++} 
\right]}E_+  &+&
{\scriptstyle \left[
(\omega_\gamma^2-q^2)- 2\omega_\gamma^2 h_{+-} 
\right]}E_ -  \,\simeq\,0\, .
\ea
We have for a symmetric and traceless $3\times 3$ matrix $X_{ab}$
\ba
&X_{-+}=\tfrac{1}{2}(X_{11}+X_{22})=-\tfrac{1}{2}X_{\hat n \hat n} \, ,\q  & X_{--}=\tfrac{1}{2}(X_{11}-X_{22}-2\imath X_{12}) \, ,\nn\\
&X_{++}=\tfrac{1}{2}(X_{11}-X_{22}+2\imath X_{12}) \,, \q & X_{+-}=\tfrac{1}{2}(X_{11}+X_{22})=-\tfrac{1}{2}X_{\hat n \hat n}\, ,
\ea
and thus 
\ba\label{EM59}
 \left[
(\omega_\gamma^2-q^2)+ \omega_\gamma^2 h_{\hat n\hat n} 
  \right] E_+\, \simeq\,
      2\omega^2_\gamma \chi^+_{--}  \,\,E_-  \, , \q\q
\left[
(\omega_\gamma^2-q^2)+ \omega_\gamma^2 h_{\hat n\hat n} 
  \right]E_- \,\simeq\,    2\omega^2_\gamma \overline{(\chi^+_{--})}  \,\,E_+
\ea
where $\chi^+_{--}= \chi^1_{--}+\imath \chi^2_{--}=\tfrac{1}{2}(\chi^1_{11}- \chi^1_{22} -2 \imath \chi^1_{12}  +\imath \chi^2_{11}- \imath \chi^2_{22} +2  \chi^2_{12} )$.  Above we have used one of the reality condition $\overline{\chi^+_{--}}=\chi^-_{++}$, which the $\chi^\pm$ have to satisfy.

Multiplying the two
equations
gives the dispersion relation (to second order)
\ba
\big(\omega_\gamma^2-q^2+\omega_\gamma^2 h_{\hat n\hat n}\big)^2
=4\omega_\gamma^4\,| \chi^+_{--}|^2
\ea
whose two roots $q_\pm$ define the two \emph{eigen-indices} $n_\pm=q_\pm/\omega_\gamma$. To first order we obtain\footnote{Here we chose to parameterize the eigen-index difference with $\chi^+$. We could have also chosen $\chi^-$ as these are connected by $\overline{\chi^+_{ab}}=\chi^-_{ab}$, and in particular $\overline{\chi^+_{--}}=\chi^-_{++}$.}
\ba\label{eq:eigenindex}
n_\pm \,=\, 1+\tfrac{1}{2} h_{\hat n \hat n} \pm | \chi^+_{--}| \, .
\ea

For each of the two roots $n_P$ with $P=\pm$, the equations \eqref{EM59} also define eigen-polarizations $(E^P_+,E^P_-)$. These are the only polarizations that propagate without changing their
polarization along their path.
Inserting $\omega_\gamma^2-q_\pm^2+\omega_\gamma^2h_{\hat n\hat n}
=\pm2\omega_\gamma^2|\chi^+_{--}|$ into the first equation of \eqref{EM59} gives
\ba\label{eq:eigenpol}
E^P_-\,=\,P e^{-\imath\arg(\chi^+_{--})}\,E^P_+\, , \q P=\pm\, .
\ea
The two circular components of an eigen-mode thus have equal magnitude, which
characterizes a linear polarization, and we have the case of pure birefringence.\footnote{This would change if we would not impose the reality conditions on $\chi^\pm$. For general complex  $\chi^\pm$ the eigen-indices can become complex and result in amplification and damping of the amplitudes.} In the linear polarization basis we have $(E_1^P,E_2^P)\propto (\cos\varphi^P, \sin\varphi^P)$ with 
\ba\label{eq:psi}
\varphi^-=\psi \q\text{and}\q \varphi^+=\psi+\pi/2,
\q\q\text{with}\q\q
\psi\,\equiv\,-\tfrac12\arg\big(\chi^+_{--}\big)\,.
\ea
The first eigen-polarization propagating with the smaller index $n_-$ (the ``fast axis''), the second with
$n_+$ (the ``slow axis'').  That is,  general $\chi^\pm$-fields break rotation invariance around $\hat n$ and lead to a preferred direction. We remind the reader, that  for the massless area metric excitations the $\chi$ fields are 'slaved' to the graviton field, so it is the gravitational wave background that determines the preferred directions.
The two eigen-polarizations are the vacuum analogue of the ordinary and extraordinary rays of a
birefringent crystal, with axes depending on the photon direction $\hat n$.

For the $\chi$-dressed gravitational wave the two observables --- the magnitude $\Delta n=n_+-n_-$ and the
axis angle $\psi$ --- take a simple form. Consider the plus-polarized wave of
\eqref{eq:chi12dress}, travelling in the $\hat z$-direction with length-metric
amplitude $h_0\equiv h_{\hat x \hat x}=A/\sqrt2$, and a photon travelling along $\hat x$, so that
the transverse frame is $(\hat e_1,\hat e_2)=(\hat y,\hat z)$. The only nonvanishing
transverse component in \eqref{eq:chi12dress} is
$\chi^1_{\hat y\hat y}=-\sqrt2\,(k^2/m^2)A$, while $\chi^2_{\hat y\hat y}=0$, and one finds
\ba\label{eq:gravbiref}
\chi^+_{--}\,=\,-\,\frac{k^2}{m^2}\,h_0
\q\q\Longrightarrow\q\q
&&\Delta n \,=\, 2\,\Big(\frac{k}{m}\Big)^{2} h_0\,,
\nn\\
&&\psi\,=\,0\ \ (\text{mod }\pi/2)\,.
\ea
Since the dressing \eqref{eq:chi12gen} does not depend on $\beta_\pm$, neither the
magnitude nor the axes of the birefringence depend on the parity parameter $\xi$: the
birefringence axes coincide with the projected polarization axes $(\hat y,\hat z)$ of the
gravitational wave for all $\xi$. Thus the dressing
of gravitational waves does not provide access to $\xi$.

 The situation is different for the massive
branch, whose polarization structure depends on $\xi$ through \eqref{eq:massivebranch}.
Consider a massive wave of helicity $H$ travelling in the $\hat z$-direction,
$\chi^{+,H}=X$, $\chi^{-,H}=-\tfrac{\beta_+Q^H_-}{\beta_-Q^H_+}X$, and again a photon
along $\hat x$. The physical configuration is the superposition of this mode and its
conjugate partner required by \eqref{eq:reality}. The self-dual field is then
$\chi^+_{ab}(x)=X\,t^H_{ab}e^{\imath\theta}+\overline{\chi^{-,H}}\;\overline{t^H_{ab}}\,e^{-\imath\theta}$, and $\chi^-_{ab}=\overline{\chi^+_{ab}}$. For the photon along $\hat x$
only the component $\chi^+_{--}=\tfrac12(\chi^+_{\hat y\hat y}-\chi^+_{\hat z\hat z}-2\imath\chi^+_{\hat y\hat z})=\tfrac12\chi^+_{\hat y\hat y}$ enters
the optical equations \eqref{EM59}. It is the complex combination $\chi^1_{--}+\imath\chi^2_{--}$ of the real fields $\chi^1,\chi^2$, whose modulus and
phase determine $\Delta n$ and $\psi$. It is not itself real. With $t^H_{\hat y\hat y}=-\tfrac12$ and
$\beta_-/\beta_+=e^{-2\imath\arctan(\sinh2\xi)}$ it evaluates to
\ba
\chi^+_{--}(t)\,=\,-\tfrac14\Big(X e^{\imath\theta}
-\frac{Q^H_-}{Q^H_+}\,e^{-2\imath\arctan(\sinh 2\xi)}\,\overline{X}e^{-\imath\theta}\Big)\,,
\q\q \theta=kz-\omega t\,.
\ea
As $\theta$ advances, $\chi^+_{--}$ traces an ellipse in the complex plane, so that the
birefringence axes of the photon rotate with the wave. The magnitude $\Delta n$ oscillates
between $\tfrac12\big|1-Q^H_-/Q^H_+\big|\,|X|$ and $\tfrac12\big(1+Q^H_-/Q^H_+\big)|X|$, and at the maxima the
axis angle takes the value
\ba\label{eq:massbiref}
\psi_{\rm max}\,=\,-\frac{\pi}{4}+\frac12\arctan\big(\sinh 2\xi\big)\,=\,-\frac{\pi}{4}+\frac12\arctan\big(\gamma^{-1}\big)\q(\text{mod }\pi/2)\, ,
\ea
independently of the phase of $X$ and of the helicity. The parity parameter thus
rotates the axes of maximal birefringence induced by massive non-length waves away from
their parity-symmetric orientation.

\subsection{Remarks on observational signature}

To explain how these effects could be observed, we briefly recall the working principle
of a gravitational-wave interferometer. The instrument consists of two orthogonal arms of
length $L$ (four kilometres for LIGO) with mirrors at their ends. The mirrors are
suspended as pendula; at the measurement frequencies, roughly $10$--$10^3\,$Hz and thus
far above the pendulum resonance, they respond to gravity as freely falling test
masses in the horizontal plane, i.e.\ they follow geodesics of the length
metric. Laser light is stored in each arm --- optical cavities let it traverse the arm
effectively a few hundred times, so the effective path is $L_{\rm eff}\sim10^3\,$km ---
and the light returning from the two arms is interfered at the output port; the
observable is the differential optical phase between the arms.

In the transverse--traceless gauge used throughout this paper, the mirrors remain at
fixed coordinates, and the entire signal resides in the optical phase. In the absence of non-length perturbations \eqref{eq:eigenindex} shows that the eigen-indices degenerate and the phase is thus polarization independent. The phase
accumulated along an arm in the direction $\hat n$ is
$\omega_\gamma L_{\rm eff}\,\big(1+\tfrac12 h_{\hat n\hat n}\big)$, so a plus-polarized
wave arriving from overhead produces the differential arm phase
\ba\label{eq:strainchannel}
\Delta\phi_{\rm strain}(t)
=\tfrac12\,\omega_\gamma L_{\rm eff}\,\big(h_{\hat x\hat x}-h_{\hat y\hat y}\big)
=\omega_\gamma L_{\rm eff}\,h_0(t)\,.
\ea
To provide some numbers,\footnote{In formulas we use units with $\hbar=c=1$, so that $m$ denotes both the mass and the
corresponding frequency $mc^2/\hbar$; numerical values of $m$ are quoted in eV, and ratios
such as $\omega_{\rm gw}/m$ stand for $\hbar\omega_{\rm gw}/(mc^2)$. The fields $h_{\mu\nu}$ and
$\chi^\pm_{ab}$ are dimensionless.}  a compact binary at
hundreds of megaparsecs produces a strain $h_0\sim10^{-21}$, and with
$\omega_\gamma L_{\rm eff}/c\sim10^{13}$ for infrared laser light this gives
$\Delta\phi_{\rm strain}\sim10^{-8}\,$rad.  In pure gravity this effect is (photon-)polarization blind:  a length metric defines a  single light cone for both photon polarizations, and $\Delta n$ vanishes identically.

The non-length degrees of freedom lead to a polarization dependent  differential arm phase. Within a
\emph{single} arm, light polarized along the two eigen-axes accumulates the relative
phase
\ba\label{eq:polchannel}
\Delta\phi_{\rm pol}(t)
=\omega_\gamma\,L_{\rm eff}\,\Delta n(t)
=2\,\omega_\gamma L_{\rm eff}\,\big|\chi^+_{--}\big|(t)\,,
\ea
oscillating at the gravitational-wave frequency and phase-locked to the strain signal
recorded simultaneously in the standard channel. This measurement is
polarimetry rather than interferometry between arms: one sends light of two orthogonal
linear polarizations through the same arm or a single beam whose polarization is
modulated) and reads out their differential phase, separated by polarizing
optics at the output. $\Delta\phi_{\rm strain}$ and $\Delta\phi_{\rm pol}$ are complementary projections of the same
optical field: the strain phase is common to both polarizations and cancels in the
polarization readout, while the birefringent phase cancels between polarizations in the
strain readout.

The polarization channel is a \emph{null test}: in general relativity it vanishes
identically, so a signal at any level indicates non-length geometry. Its practical
signatures distinguish it sharply from instrumental effects. Mirror coatings and
substrates are themselves birefringent at levels vastly exceeding $\Delta n$, but that
birefringence is static; the signal, by contrast, (a) oscillates at the
gravitational-wave frequency, (b) is phase-locked to the simultaneously recorded
strain, (c) carries a spectrum tilted by two powers of gravitational wave frequency relative to the
strain, $\Delta n/h_0=2\,(\omega_{\rm gw}/m)^2$, and (d) depends on the
orientation of the light polarization: rotating the input polarization by an angle
$\vartheta$ modulates the signal as $\cos2(\vartheta-\psi)$, with the axis angle
$\psi$ aligned with the wave's polarization axes known from the strain channel,
cf.~\eqref{eq:gravbiref}.

An instrumental precedent
for the readout is provided by axion dark-matter searches via polarization rotation in
gravitational-wave interferometers \cite{Nagano:2019rbw,Oshima:2023csb}, there for a scalar, quasi-monochromatic background producing circular birefringence, here for a tensorial, transient, strain-locked \emph{linear}
birefringence.

We point out two caveats: First, the magnitude: the ratio of the two
channels is $\Delta\phi_{\rm pol}/\Delta\phi_{\rm strain}
=2\,(\omega_{\rm gw}/m)^2$, which for a Planckian mass scale $m$ and
$\omega_{\rm gw}\sim10^2\,$Hz is of order $10^{-81}$.  This channel is therefore more of a matter of
principle unless $m$ lies far below the Planck scale.  But 
a non-detection bounds $m$  independently of astrophysical
modelling. Second, the mirrors are electromagnetically bound solids, so a
$\chi$-background also perturbs atomic bonds and hence the mirror dimensions, feeding a
strain-like signal at the same $(\omega_{\rm gw}/m)^2$ order; this contribution is
material-dependent and would add to the polarization signature.

Concerning massive non-length waves their signal would be outside available frequency bands for a Planck scale mass. The signal would become meaningful only for an ultralight
non-length sector, $m\lesssim \hbar \omega_{\rm gw}\sim10^{-12}\,$eV, with waves near the threshold frequency being very slow and strongly dispersive.  A condensate of $\chi$, oscillating at $\omega\simeq m$, would
produce a narrow-line birefringence at the frequency $m/2\pi\hbar$, structurally analogous to
axion dark-matter searches with laser interferometers but in the polarization channel. Such a
condensate is not an incoherent superposition of waves from all directions. Within a
coherence length $\lambda_c\sim2\pi/(mv)$ and a coherence time $\tau_c\sim2\pi/(mv^2)$, where $v$ is the
velocity dispersion in units of $c$ ($v\sim10^{-3}$ for a virialized population), it is a single coherent
oscillation with a definite polarization tensor. The birefringence is therefore present
locally, with axes set by this tensor, and oscillates at frequency $m$. It averages to zero
only over times long compared to $1/m$ without demodulation, or over many coherence patches.
A lock-in at frequency $m$ recovers it within a coherence time, with the usual
$(T\tau_c)^{1/4}$ growth of the sensitivity beyond it. Photons from distant sources crossing many
patches accumulate a random walk with vanishing mean, which shows up only as
depolarization. Massive non-length waves are the ones that carry the parity parameter: for
a wave with a definite direction of propagation the orientation of maximal birefringence is
rotated by $\tfrac12\arctan(\gamma^{-1})$, cf.~\eqref{eq:massbiref}. For a condensate at rest the same
relative phase between $\chi^+$ and $\chi^-$ is present for all polarizations. However, for a
single photon direction it is degenerate with the unknown phase structure of the condensate
tensor, so that its measurement would require several photon directions or sensors within
one coherence patch. A
polarization-instrumented interferometer is thus, in principle, an antenna for a class
of signals invisible to every standard gravitational-wave search.

Atom interferometers such as AION \cite{Badurina:2019hst} and MAGIS-100 \cite{MAGIS-100:2021etm} could be used to detect the massive non-length branch. Such  interferometers
use freely falling strontium atoms as test masses and compare the laser phases imprinted on
atom clouds separated by baselines of $10\,$m to $1\,$km, in the mid-frequency band of about
$0.01$--$10\,$Hz.   A background
of massive non-length waves, or a condensate oscillating at $\omega\simeq m$, produces a
persistent, narrow-line and phase-coherent birefringence. The frequency band of AION
corresponds to masses $m\sim10^{-17}$--$10^{-14}\,$eV, the ultralight dark matter window for which
these instruments were designed. It is this channel that carries the parity parameter. For waves with a definite direction the
axes of maximal birefringence are rotated by $\tfrac12\arctan(\gamma^{-1})$ relative to their
parity-symmetric orientation, cf.~\eqref{eq:massbiref}. For a condensate the parameter is
accessible only by comparing several photon directions within a coherence patch, see above.

As a word of caution, an
order-one coupling of the Maxwell field to the area metric at such small masses has to be
confronted with tests of the equivalence principle and with bounds on vacuum birefringence
from polarized astrophysical sources. One should also keep in mind that the optical analysis of
Section~\ref{sec:EM} rests on a number of approximations: (a) the geometric optics assumption $\omega_\gamma\gg\omega_{\rm wave},|\vec k_{\rm wave}|$ and (b) the assumption of a frozen area metric during the passage of light: $\omega_{\rm wave}T\ll1$, where $T$ is
the effective light travel time (the storage time of an arm cavity, or the interrogation time
of an atom interferometer). Furthermore we work here (c) within a perturbative set-up. A very light mass enhances the polarization channel relative to the strain channel due to the dressing $|\chi|\sim(\omega_{\rm gw}/m)^2h_0$. But the
linearized treatment breaks down for $m\sim\omega_{\rm gw}\sqrt{h_0}$, that is for $|\chi|\sim 1$.


\section{Discussion and outlook}\label{sec:discussion}

In this work we developed the Lagrangian analysis of linearized, shift-symmetric
area-metric gravity in Lorentzian signature, which can include parity violating terms. The area metric excitations can be split into length metric and non-length excitations.  The plane-wave solutions organize into two
branches: an exactly length metric massless graviton, dressed by non-length excitations at relative
order $(k_{\text{grav}}/m)^2$, and a massive branch with
$h_{ab}\equiv0$, which carries strictly negative energy gapped at the (presumably
Planckian) mass $m$.  With the shift-symmetry condition imposed, the length metric excitations propagate in the same way as in (linearized) general relativity, that is without velocity or amplitude birefringence. This is different from most other parity violating gravitational theories \cite{Qiao:2022mln}, where the gravitational waves themselves feature velocity or even amplitude birefringence.

Coupling the theory to
electro-magnetism, the length perturbation produces the
familiar polarization-blind phase (the strain channel, unmodified at linear order),
while the non-length perturbations render the vacuum birefringent. For the
non-length dressing of a gravitational wave the eigen-axes are aligned with the
polarization axes of the wave and the magnitude is independent of $\xi$,
cf.~\eqref{eq:gravbiref}. The parity parameter enters instead the polarization
structure of the massive non-length waves, whose axes of maximal birefringence are
rotated by $\tfrac12\arctan(\sinh2\xi)=\tfrac12\arctan(\gamma^{-1})$,
cf.~\eqref{eq:massbiref}. Through $\sinh(2\xi)=\gamma^{-1}$ the polarization channel
is therefore, in principle, a measurement of the Barbero-Immirzi parameter, provided
massive non-length waves are present --- although for Planckian $m$ such signals are
extremely weak, and a realistic target only if $m$ lies far below the Planck scale.

The extreme smallness of this parity violation effect leaves however room for a small Barbero-Immirzi parameter. (Remember that $\sinh(2\xi)=\gamma^{-1}$.) In loop quantum gravity and spin foams the Barbero-Immirzi parameter acts also as a anomaly parameter \cite{Dittrich:2012rj}, one therefore expect rather a small value \cite{Magliaro:2011zz,Han:2013ina,EffSFE,Asante:2020iwm}. Arguments from black hole counting in loop quantum gravity \cite{BarberoG:2022ixy,Meissner:2004ju} suggest $\gamma=0.237\cdots$ and analysis of the renormalization flow suggest  that parity symmetric values for the couplings are IR repulsive, see \cite{Daum:2010qt,Benedetti:2011nd,Daum:2013fu} for earlier works within first order gravity and \cite{Borissova:2025frj} for a study of area metric gravity.

The next step is the extension of (shift symmetric) area metric actions to non-linear order and to non-flat backgrounds. This can be accomplished via a construction of area metric actions from the modified Plebanski framework, extending  \cite{Borissova:2022clg} to the non-linear theory \cite{ToAppear}.  One can thus clarify the fate of the negative-energy massive branch.  

On a cosmological background the flat-space cancellation mechanism for parity odd terms in the effective length metric action likely fails and one expects  area-metric dynamics to generate an
amplitude asymmetry of order $\sinh(2\xi)(H_{\text{inf}}/m)^2$, where $H_{\text{inf}}$ denotes the (nearly constant) Hubble rate during
inflation, at which the physical frequencies of the tensor modes cross the horizon.  

This is complementary to the analysis in \cite{Bianchi:2024mrt}, which works with an effective length metric action including a parity even Gauss-Bonnet and a parity odd Chern-Simons term (that is higher derivative terms) and needs to also adopt a rolling scalar.  Assuming so-called $\gamma$-duality\footnote{\cite{Bianchi:2024mrt}  parametrize the ratio between the $B\wedge F$ and $\epsilon B\wedge F$ terms in the Plebanski action (\ref{PA1}) via a rotation angle $\theta$, so that $\tan(\theta)=\gamma^{-1}$ and the entire action is multiplied by $\cos(\theta)$. We think that a parametrization via a hyperbolic angle as in $\sinh(2\xi)=\gamma^{-1}$ is more appropriate in Lorentzian signature, due to the relations (\ref{eq:shiftcond}). This hyperbolic encoding results from the $\pm \imath$ eigenvalues for the Hodge star operator. Otherwise the $\gamma$-duality of \cite{Bianchi:2024mrt} and the shift symmetry condition seem to be closely related, both conditions assume an action where the parity violating terms are ruled by one parameter, $\gamma$.} imposes a relation between the Gauss-Bonnet and Chern-Simons terms, which leads to a circular-polarization asymmetry of primordial gravitational waves, which depends on $\gamma$ but also on slow-roll parameters.

In summary the enlargement of the gravitational configuration
space from length to area metrics converts a quantization ambiguity of quantum gravity
--- the Barbero-Immirzi parameter --- into a coupling of the linearized dynamics that
is, at least in principle, observable.  Different from other parity violating versions of gravity \cite{Qiao:2022mln}, shift symmetric area metric gravity does not affect the propagation of gravitational waves on a flat background.  Assuming Planck mass for the non-length excitations gives an extremely  small observable effect, but also allows for a small Barbero-Immirzi parameter, as argued for in loop quantum gravity.

\begin{acknowledgments}
	
B.~D.~thanks Johanna Borissova and Kirill Krasnov for collaboration on \cite{BDK}, which inspired the current work. B.~D.~also thanks Seth Asante, Jose Padua Arguelles, Aldo Riello for discussions, Abhay Ashtekar for correspondence, and Ann Nelson for suggesting to consider the AION experiment.
Research at Perimeter Institute is supported in part by the Government of Canada through the Department of Innovation, Science and Economic Development Canada and by the Province of Ontario through the Ministry of Colleges and Universities.\end{acknowledgments}

\bibliographystyle{jhep}
\bibliography{references}

\end{document}